\documentclass[12pt]{amsart}
\usepackage{latexsym,amssymb,amsmath,amsthm}
\usepackage{comment}
\usepackage{listings}
\usepackage{float}
\usepackage{graphicx}
\newtheorem{thm}{Theorem}
\newtheorem{lemma}[thm]{Lemma}

\theoremstyle{definition}

\def\tr{\rm{tr}}

\def\B#1#2{{#1\choose #2}}

\title{Characteristic Length and Clustering}

\author{Oliver Knill}
\date{October 12, 2014}
\address{
        Department of Mathematics \\
        Harvard University \\
        Cambridge, MA, 02138
        }


\begin{document}
\begin{abstract}
We explore relations between various variational problems for graphs: among the functionals 
considered are Euler characteristic $\chi(G)$, characteristic length $\mu(G)$, mean clustering $\nu(G)$, 
inductive dimension $\iota(G)$, edge density $\epsilon(G)$, scale measure $\sigma(G)$,
Hilbert action $\eta(G)$ and spectral complexity $\xi(G)$. A new insight in this note
is that the local cluster coefficient $C(x)$ in a finite simple graph can be written 
as a relative characteristic length $L(x)$ of the unit sphere $S(x)$ within the unit ball 
$B(x)$ of a vertex. This relation $L(x) = 2-C(x)$ will allow to study clustering 
in more general metric spaces like Riemannian manifolds or fractals. 
If $\eta$ is the average of scalar curvature $s(x)$, a formula $\mu \sim 1+\log(\epsilon)/\log(\eta)$ 
of Newman, Watts and Strogatz \cite{NewmanStrogatzWatts} relates $\mu$ with the edge 
density $\epsilon$ and average scalar curvature $
\eta$ telling that large curvature correlates with small characteristic length.
Experiments show that the statistical relation $\mu \sim \log(1/\nu)$ holds for random or 
deterministic constructed networks, indicating that small clustering is often associated to large 
characteristic lengths and $\lambda=\mu/\log(\nu)$ can converge in some graph limits of networks.
Mean clustering $\nu$, edge density $\epsilon$ and curvature average $\eta$ therefore can relate 
with characteristic length $\mu$ on a statistical level.
We also discovered experimentally that inductive dimension $\iota$ and cluster-length ratio 
$\lambda$ correlate strongly on Erd\"os-Renyi probability spaces. 
\end{abstract} 
\maketitle

\section{Introduction}

The interplay between global and local properties often appears in 
geometry. As an example, the Euler characteristic can by Gauss-Bonnet be written 
as an average of local curvature and by Poincar\'e-Hopf as an average of local index density $i_f$
for Morse functions $f$. Given a globally defined quantity on a geometric space, one
can ask to which extent the functional can be described as an average over local properties.
Also of interest is the relation between the various functionals. 
We will look at some examples on the category of finite simple graphs and comment on both 
problems. We look then primarily at characteristic length $\mu$, which is the expectation of 
local mean distance $\mu(x)$ on a metric space equipped with a probability measure. 
On finite simple graphs one has a natural geodesic distance and a natural counting measure so that
networks allow geometric experimentation on small geometries. Most functionals are 
interesting also for Riemannian manifolds, where finding explicit formulas for the characteristic length
can already lead to challenging integrals. Characteristic length is sometimes also defined
as the statistical median of $\mu(x)$ \cite{SmallWorld}. Since the difference is not essential, we use the averages
$$ \mu(G) = \frac{1}{n^2-n} \sum_{x \neq y \in V^2} d(x,y) \; , 
   \nu(G) = \frac{1}{n} \sum_{x \in V} \frac{2e(x)}{n(x)(n(x)-1)}  \; , $$
where $n(x),e(x)$ are the number of vertices and edges in the sphere $S(x)$ of the vertex $x$. 
$\mu$ is the average of the non-local quantity $D(x)=\frac{1}{n-1} \sum_{y \neq x} d(x,y)$ and 
$\nu$ is the average of the local quantity $C(x) = |e(x)|/B(n(x),2)$ giving the fraction of
connections in the unit sphere in comparison to all possible pairs in the unit sphere. \\

This averaging convention for $\mu$ is common. \cite{DoyleGraver1} showed already that
$\mu(G) \leq {\rm diam}(G)-1/2$. An other notion is the variance $v(G) = \max_x d(x) - \min_x d(x)$
where $d(x)=\sum_y d(x,y)$ for which Ore has shown that on graphs with $n$ vertices
has the maximum taken on trees. \cite{EntringerJacksonSnyder}. 
The definition of $\nu(G)$ is to average the edge density of the sphere relative to the case 
when the sphere is the complete graph with $n$ vertices in which case the number 
of edges is $n(x)(n(x)-1)/2$. Both quantities are natural functionals. The
mean cluster coefficient $\nu$ is by definition an average of local quantities. 
While it is impossible to find a local quantity whose average captures
characteristic length exactly, there are notions which come close. We look at three such 
relations. The first is a formula of Newman, Watts and Strogatz \cite{NewmanStrogatzWatts} 
which writes $\mu$ as a diffraction coefficient divided by scalar curvature.
This is intuitive already for spheres, where the signal speed and the curvature determines
the characteristic length. Empirically, we find an other quantity which also often allows to estimate
characteristic length well: it is the mean cluster density $\nu$, the average of a 
local cluster density $C(x)$ as defined by Watts and Strogatz. Thirdly, we will report on 
some experiments which correlate the length-cluster coefficient $\lambda = \mu/\log(1/\nu)$ 
with the inductive dimension $\iota$ of the graph. \\

\begin{figure}{
\scalebox{0.2}{\includegraphics{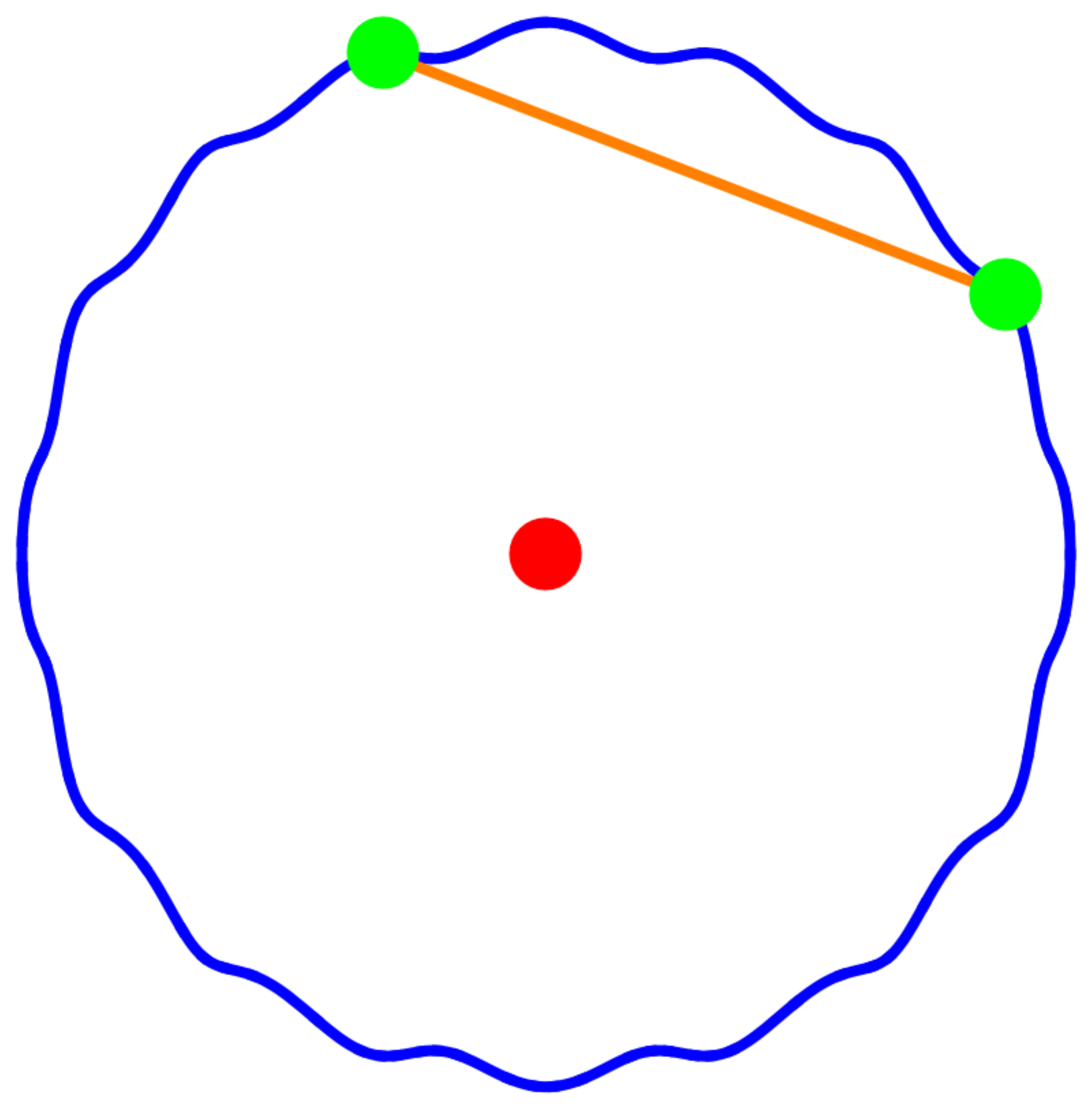}}
\scalebox{0.2}{\includegraphics{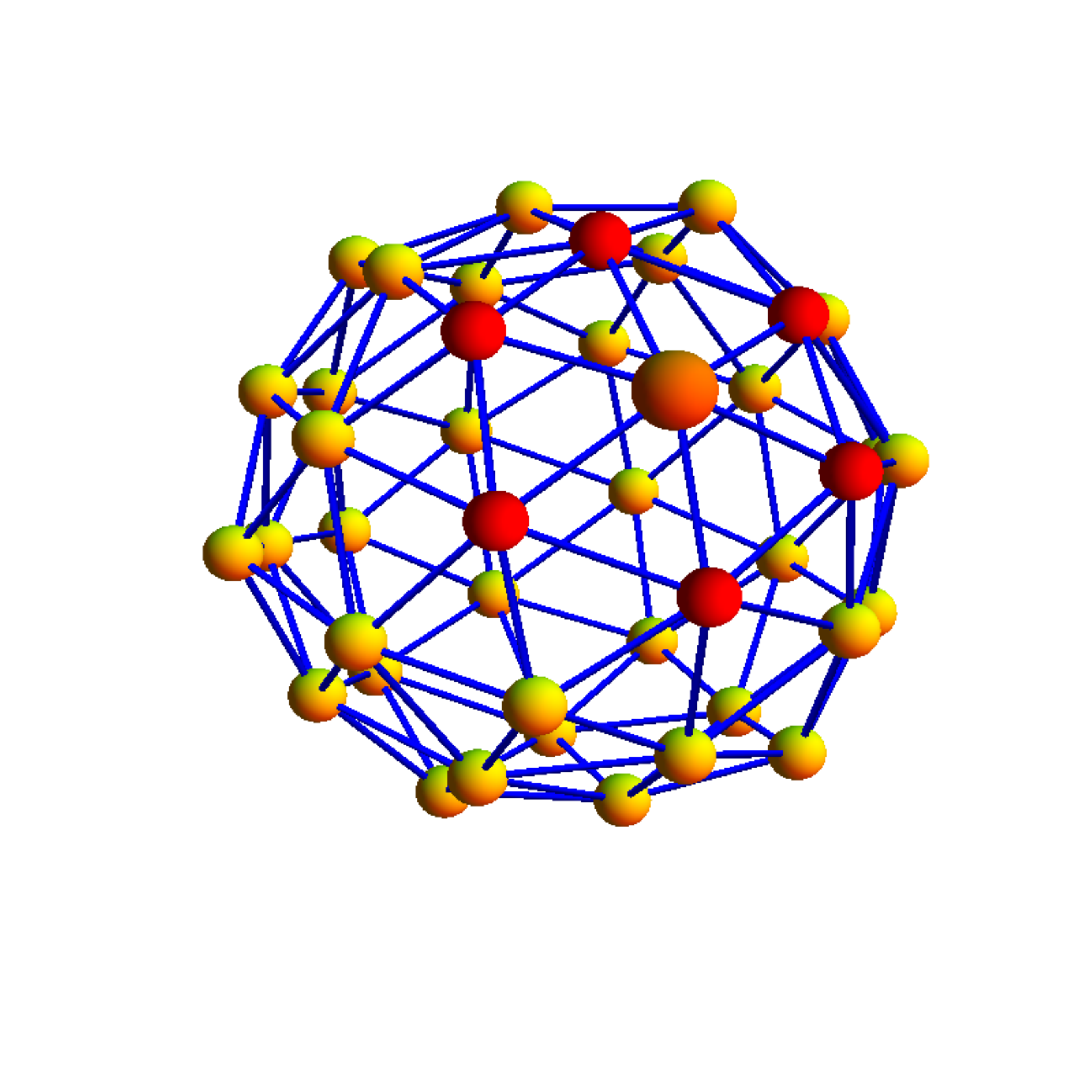}}
\caption{
For graphs, the clustering coefficient $C(x)$ at a vertex $x$ is related to the relative 
characteristic length $L(x)$ of the unit sphere $S(x)$ within the unit ball $B(x)$. 
In other words, $L(x)$ is the average distance between two points in $S(x)$ within $B(x)$. 
The relation $C(x) = 2-L(x)$ allows so to define clustering in any metric space equipped
with a measure inducing conditional probability measures on spheres. It is 
a local quantity which is constant in the radius $r$ of the spheres if we are in 
Euclidean spaces or fractal spaces with some uniformity.
The second figure shows a graph for which the local clustering coefficients $C(x)$ take the 
values $1/2$ or $2/5$, where the global characteristic length is $\mu=2.9$, 
the mean clustering coefficient is $\nu=3/7$ and
the length-cluster coefficient $\lambda=-\mu/\log(\nu)$ is $3.4$. }
\label{clustering}
}
\end{figure}

The starting point of this note is the observation that the local cluster property $C(x)$ if a vertex
in a graph can be written in terms of relative characteristic length of the unit sphere $S(x)$ within the 
unit ball $B(x)$ . We are not aware that this has been noted already, but it is remarkable as it allows 
to carry over the definition of ``local cluster property" to metric spaces equipped with a probability measure 
as long as the measure has the property that it induces probability measures $m(\cdot,S)$ on spheres by 
conditional probability $m(A) = \lim_{\epsilon \to 0} m(N_{\epsilon}(A \cap S))/m(N_{\epsilon}(S))$, where
$N_{\epsilon}(Y)$ is an $\epsilon$ neighborhood of $Y$. For such metric spaces, one can also define a
cousin of scalar curvature $\log(2^{\iota(x)-1} \delta/\delta_2)$ by comparing the measures $\delta,\delta_2$ 
of spheres of radius $1$ and $2$. We call it 
``scalar curvature" because if $\iota(x)$ is the inductive dimension of space at that point, then 
$\log(2^{\iota(x)-1} \delta/\delta_2)$ is zero if the volume $\delta_2(x)$ of the 
sphere of radius $2$ is equal to $2^{{\rm dim}(S(x))}= 2^{\iota(x)-1}$ times the volume $\delta$ 
of a sphere of radius $1$. Comparing with Riemannian manifolds, the comparison of spheres $S_{2r}(x)$
and $S_{r}(x)$ allows to measure scalar curvature.  The dimension scaling factor $\log(2)$ does not matter 
much because the Hilbert action $\eta(G)$, the expectation of the scalar 
curvature is $\log(2) \iota(G)$ plus the expectation of $s(x)=\log(\delta/\delta_2)$. Since the later is a local
function and appears also in a formula of Newman, Watts and Strogatz and because we look at the 
dimension $\iota(G)$ anyway, it does not matter if for simplicity $s(x)$ is called the
{\bf scalar curvature} and its average over vertex set is called the {\bf Hilbert action} $\eta(G)$
of the graph.  \\

Having these functionals, one can now study the relation between
local cluster property, characteristic length, dimension and curvature on rather general metric spaces
equipped with a natural measure. But first of all these functionals could be used to select ``natural
geometries".  Functionals are important in physics because most fundamental laws are of variational nature. 
In a geometric setup and especially in graph theory, one can consider the Euler characteristic, 
the characteristic length, the Hilbert action $\eta(G)$ given as an average scalar curvature or 
the spectral complexity $\xi(G)$, given as the product of the nonzero eigenvalues of the Laplacian $L$ 
of $G$. An other related quantity is the number of rooted spanning forests in $G$ which is 
$\theta(G) = \det(1+L)$ by the Chebotarev-Shamis theorem
\cite{ChebotarevShamis2,ChebotarevShamis1,Knillforest,CauchyBinetKnill}.
Many other extremal problems are studied in graph theory \cite{BollobasExtremal}. \\

For finite simple graphs, the Euler characteristic $\chi(G) = v_0-v_1+v_2-...$ 
with $v_k$ is the number of $k$ dimensional simplices $K_{k+1}$ of $G$. In geometric
situations like four dimensional geometric graphs, for which the spheres are discrete three dimensional 
spheres, this number can be seen as a quantized Hilbert action \cite{eveneuler} because it is an 
average of the Euler characteristic over all two-dimensional subgraphs and then via Gauss-Bonnet 
an average over a set of sectional curvatures and all points. 
The characteristic length is the average length of a path between two points relating 
to the other variational problem in general relativity. 
The Hilbert action itself is an average of scalar curvature, which can be 
defined for a large class of metric spaces. Spectral complexity
is natural because of the matrix tree theorem of Kirkhoff (see i.e. \cite{Biggs} )
relates it with the number of trees in a graph.  \\

While the Euler characteristic is an average of curvature by Gauss-Bonnet both in the Riemannian 
and graph case \cite{cherngaussbonnet} and Hilbert action is an average of scalar curvature, 
both the characteristic length $\mu$ as well as the complexity $\xi$ can not be written 
as an average of local properties: look at two disjoint graphs connected along a 
one-dimensional line graph. Cutting that line in the middle can change both quantities in a very different way 
depending on the two components which are obtained. If $\mu$
or $\xi$ were local, the functional would change by a definite value, independent of the length of the ``rope".
The characteristic length has been noted to be relevant for molecules: the chemist Harry Wiener found
correlations between the {\bf Wiener index} $W(G)$, the sum over all distances which is 
$W(G) = n(n-1) \mu(G)$ and the melting point of a hydrocarbon $G$. As reported in 
\cite{GoddardOellermann}, the characteristic length has first been considered in 
\cite{DoyleGraver1} in graph theory. For spectral properties, see \cite{Miegham}. 
The observation that the Wiener index satisfies $W(G) \leq W(T)$ 
for any spanning tree $T$ of $G$ was made in \cite{Schmuck}.  \\

The characteristic length $\mu(G)$ has been studied quite a bit. There are few classes of
graphs, where one can compute the number explicitly: for complete graphs, we have 
$\mu(K_n)=1$ for complete bipartite graphs $\mu(K_{a,b} = (2a^2+2b^2+ab)/((a+b)(a+b-1))$
for line graphs $\mu(L_n) = (n+1)/3$, for cyclic graphs $\mu(C_n) = (n+1)/4$.
for odd $n$ and $n^2/(4 (n-1))$ if $n$ is even \cite{EntringerJacksonSnyder}.
No general relation between length and diameter exists besides the trivial $\mu(G) \leq {\rm diam}(G)$. 
We have $1 \leq \mu(G) \leq (n+1)/3$ \cite{DoyleGraver1},
On the class $G(n,m)$ of graphs with $n$ vertices and $m$ edges
one has $\mu(G) \leq 2 - (2m)/(n(n-1))$ \cite{EntringerJacksonSnyder}.
Among all graphs with $n$ vertices the maximum $(n+1)/3$ is obtained for line graphs, the minimum $1$ 
for complete graphs \cite{DoyleGraver1}. The problem to find the maximum among all graphs of given 
order and diameter is unknown. There are relations with the spectrum: $\mu(G) \leq b-(2(b-1)m/(n(n-1))$ on $G(n,m)$
where $b$ is the number of distinct Laplacian eigenvalues. There are also upper bounds in terms of the second
eigenvalue.  
For a connected graph $\mu(G) \leq \beta(G)$ where $\beta(G)$ is the independence number \cite{Chung1988},
the maximal number of pairwise nonadjacent vertices in $G$. 
On all graphs of order $n$ and minimal degree $\delta$, then $\mu(G) \leq n/(\delta + 1) + 2$ \cite{KouiderWinkler}.
The conjecture generating computer program Graffiti \cite{Fajtlowicz}
suggested for constant degree $\delta$ graphs to have $\mu(G) \leq n/\delta$.
There is a spectral relation$\mu \geq \tr(L^+) 2/(n-1)$, where
$L^+$ is the Moore pseudo inverse of the Laplacian $L$ and $n$ the number of vertices( \cite{Sivasubramanian}
which is the McKay equality for trees and otherwise always a strict inequality.  \\

The Euler characteristic is definitely one of most important functionals in geometry if not the 
most important one. It is a homotopy invariant and can by Poincar\'e-Hopf be expressed cohomologically 
as $\sum_{i=0} (-1)^i b_i$ using $b_i$ the dimensions of cohomology groups $H^i(G)$. 
A general theme in topology is to extend the notion of Euler characteristic to larger classes of 
topological spaces. This essentially boils down to the construction of cohomology. 
The limitations are clear already in simple cases like the Cantor set which  have
infinite Euler characteristic as half of the space is homeomorphic to itself. Euler characteristic 
can be defined for a metric space if there exists a subbasis of contractible graphs. The 
Euler characteristic can then be defined as the Euler characteristic of the nerve graph. This
illustrates already how important homotopy is in general when studying Euler characteristic.
It is also historically remarkable that the first works done by Euler on Euler characteristic 
were of homotopy nature by deforming the graph.  \\

Various other functionals have been considered on graphs. The {\bf average centrality} $f(G)$
is the mean of the {\bf local closeness centrality}
$$ f(x) = \sum_{y} \frac{1}{\sum_{y \neq x} d(x,y)} $$
of a vertex $x$. An other number is the {\bf geodetic number} $g(G)$ which is the minimum
cardinality of a geodetic set in $G$, where a set is called {\bf geodetic} if its geodesic closure
is $G$ \cite{BKT}. The {\bf scale measure} of a graph is defined as $\sigma(G) = s(G)/m$, where 
$s(G) = \sum_{e \in E} d(e)$ and $d((a,b)) = d(a) d(b)$ and $m={\rm max}_{e \in E} d(e)$.  
An other important notion is the {\bf chromatic number} $c(G)$ which can be seen as the smallest
$p$ for which a scalar function with values in $Z_p$ exists for which the gradient 
field $df$ nowhere vanishes. Related to graph coloring, we have in \cite{colorcurvature} defined 
{\bf chromatic richness} $C(c)/c!$ measuring the size of the set of coloring functions modulo permutations.
The {\bf arboricity} $a(G)$ of $G$ is the minimal number of spanning forests 
which are needed to cover all edges of $G$. One knows that $\theta(G) = \det(L+1)$ by 
Chebotarev-Shamis \cite{ChebotarevShamis2,ChebotarevShamis1,Knillforest,CauchyBinetKnill}.
While the number of spanning forests is a measure of complexity, the arboricity is a measure of 
``denseness" of the graph. The {\bf Nash-Williams formula} \cite{NashWilliams,CMWZZ} 
tells that the arboricity is the maximum of $[m_H/(n_H-1)]$, where $n_H$
is the number of vertices and $m_H$ the number of edges of a subgraph $H$ of $G$ and where $[r]$ is
the ceiling function giving the minimum of all integers larger or equal than $r$. For example, 
for $K_{4,4}$ where $m=16, n=8$ the arboricity must be at least $16/7=2.28$ and so at least $3$ 
and one can give examples of three forests covering all edges. 
For a complete graph, the arboricity is $[n/2]$. The arboricity gives a bound on 
the chromatic number $c(G) \leq 2 a(G)$ (a fact noted in \cite{ButlerArboricity} and follows 
from the fact that each forest can be colored by 2 colors). 
The {\bf Laplacian ratio} $p(G) = {\rm per}(L)/\prod_i d_i$, 
where ${\rm per}(L)$ is the 
permanent of the Laplacian and $d_i$ are the vertex degrees has been introduced in \cite{BrualdiGoldwasser}.
The {\bf symmetry grade} of a graph is the order of the automorphism group of $G$. For the complete graph 
for example, it is $n!$ while for a cyclic graph it is $2n$, the size of the dihedral group. 
The {\bf domatic number} $d(G)$ of a graph finally is the maximal size of a dominating partition of 
the vertex set.  \\

\begin{center}
\begin{tabular}{|l|l|l|l|l|} \hline
Functional               &        &  Based  on          &   Local  &  Spectral         \\ \hline 
Euler characteristic     & $\chi$ &  Curvature          &   yes    &  somehow \cite{knillmckeansinger}  \\ 
Inductive dimension      & $\iota$&  Point dimension    &   yes    &  not known        \\
Characteristic length    & $\mu$  &  Distance           &   no     &  on trees \cite{GoddardOellermann}  \\ 
Complexity               & $\xi$  &  Eigenvalues        &   no     &  yes              \\ 
Forest complexity        & $\theta$ &  Eigenvalues      &   no     &  yes              \\
Hilbert action           & $\eta$ &  Scalar curvature   &   yes    &  not known        \\
Mean cluster             & $\nu$  &  Local cluster      &   yes    &  yes              \\ 
Average degree           &$\delta$&  Vertex degree      &   yes    &  yes              \\ 
Graph density            &$\epsilon$& Edge number       &   yes    &  yes              \\
Scale measure            &$\sigma$&  Vertex degree      &   yes    &  not known        \\
Cluster-length-ratio     &$\lambda$& Distances          &   no     &  not known        \\
Independence number      & $\beta$ & Adjacency          &   no     &  not known        \\
Variance                 & $v$    &  Distance           &   no     &  not known        \\
Centrality               & $f$    &  Local centrality   &   yes    &  not known        \\
Chromatic number         & $c$    &  Gradient fields    &   no     &  no \cite{Cvetkovic} \\
Arboricity               & $a$    &  Forests            &   no     &  not known        \\
Geodetic number          & $g$    &  Geodesics          &   no     &  not known        \\ 
Domatic number           & $d$    &  Partitions         &   no     &  not known        \\
Symmetry grade           & $t$    &  Symmetry group     &   no     &  not known        \\
Laplacian ratio          & $p$    &  Permanent          &   no     &  not known        \\ \hline
\end{tabular}
\end{center}

Besides the question whether a functional is local,
it would also be interesting to know more about which properties are spectral properties. 
Euler characteristic can be seen as a spectral property {\bf in the wider sense}:  
it is the super trace of $e^{-t D^2}$ for the Dirac operator $D$ for every $t$ by 
McKean-Singer \cite{knillmckeansinger}. 
The average degree $\delta$ can be written in terms of the adjacency matrix $A$ as 
$\delta = 2 {\rm tr}(A^2)/{\rm tr}(A^0)$ and
with the Laplacian $L$ as ${\rm tr}(L)/{\rm tr}(L^0)$. The {\bf graph density}
$\epsilon=\delta/(n-1) = 2 v_1/(v_0(v_0-1))$ is also spectral, because both $\delta$ and $n=v_0$ are spectral.

\begin{figure}
\parbox{15.4cm}{
\parbox{15cm}{\scalebox{0.42}{\includegraphics{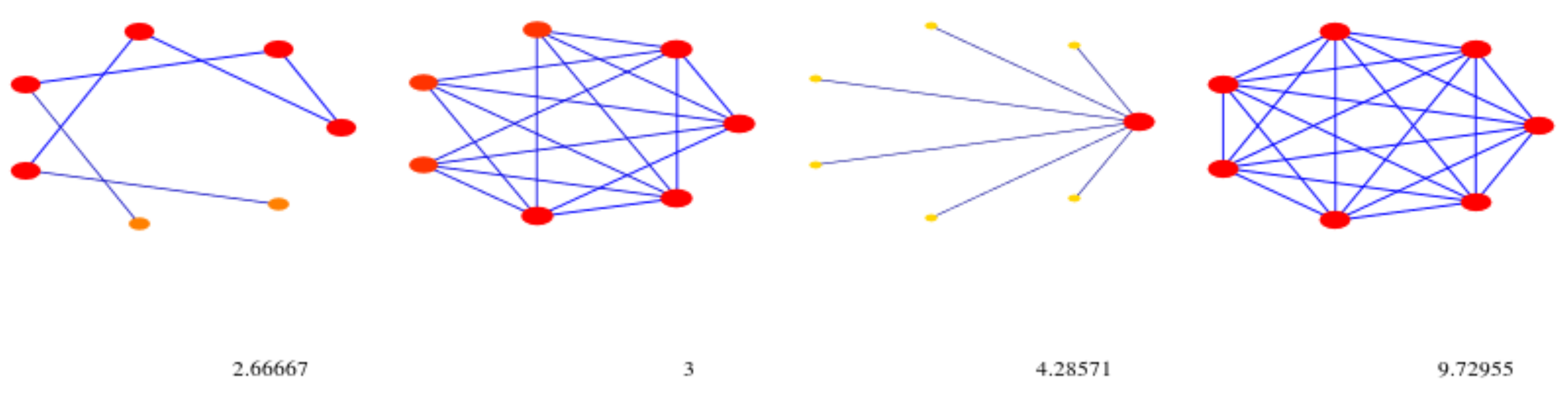}}}
\parbox{15cm}{\scalebox{0.42}{\includegraphics{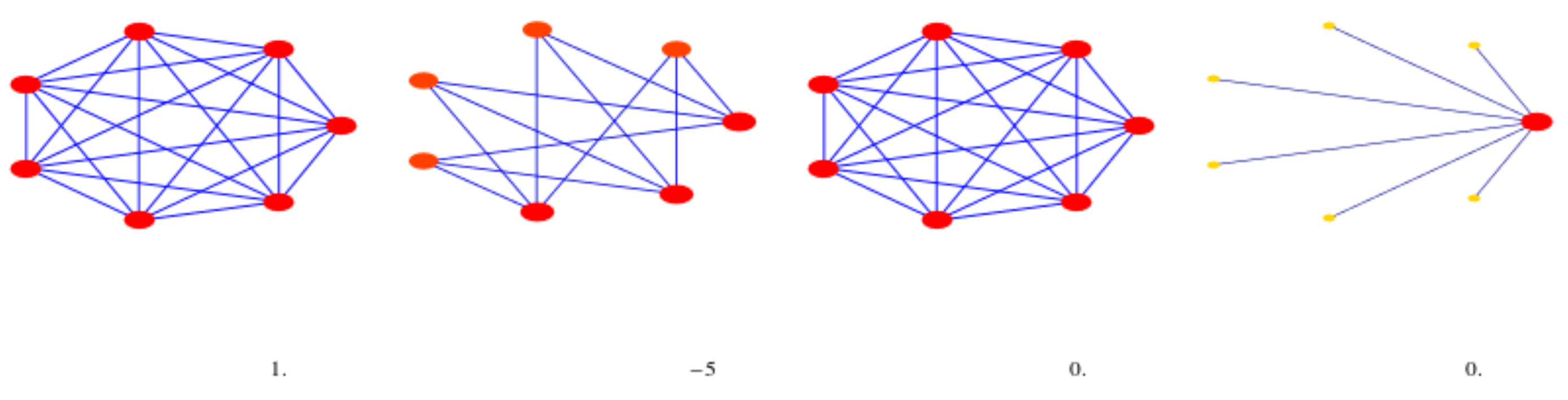}}}
\parbox{15cm}{\scalebox{0.42}{\includegraphics{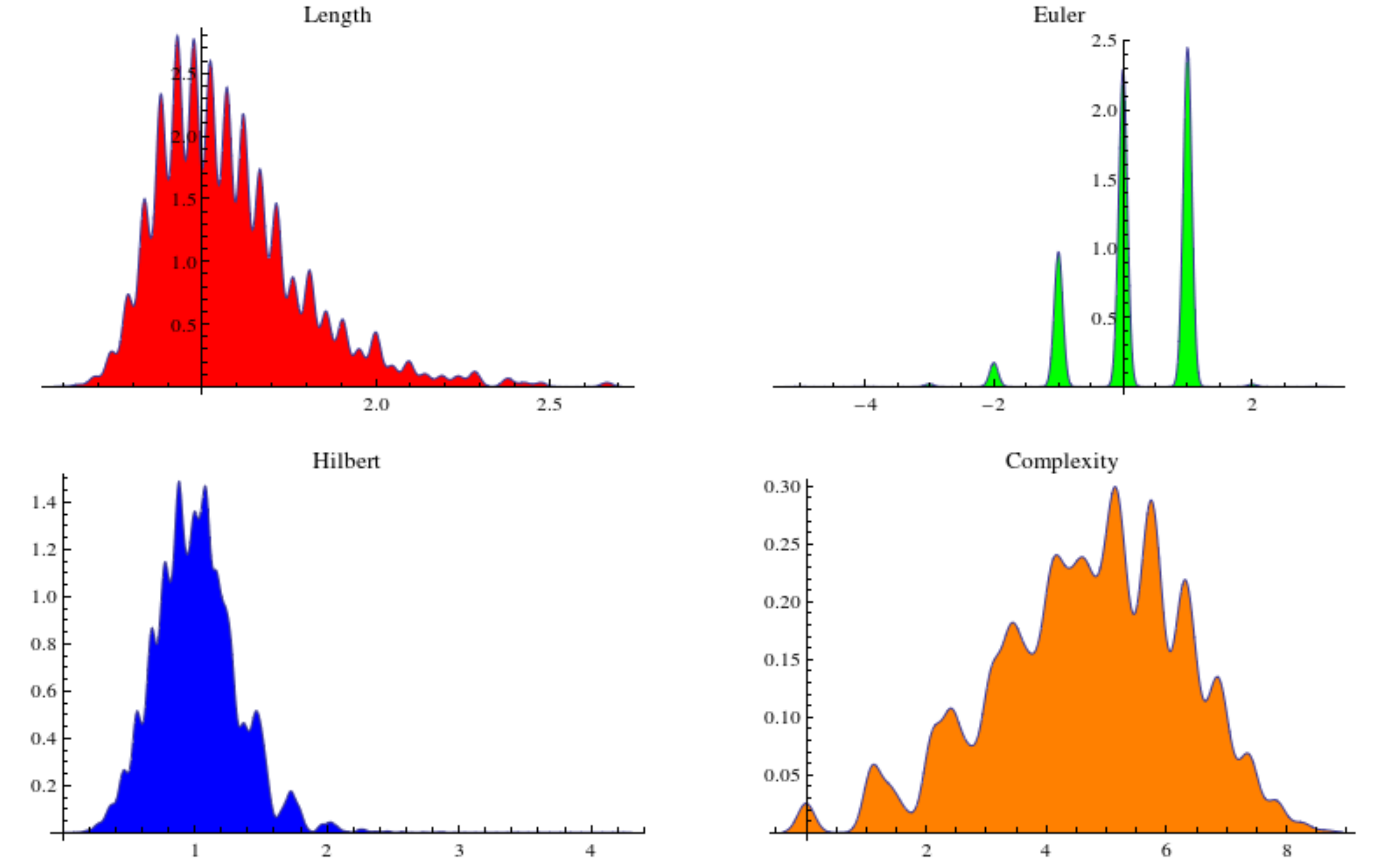}}}
}
\caption{
Graphs with minimal and maximal characteristic length, 
Euler characteristic, Hilbert action and the logarithm of the complexity among graphs all
connected graphs with $7$ vertices. We see then the distribution of the four functionals
on this finite probability space with $1'866'256$ elements. The figure illustrates that
Euler characteristic is the most interesting functional. All others appear to be extremal for
line graphs, complete graphs or star graphs, all three of which do not carry 
interesting geometries. 
}
\label{extrema}
\end{figure}

\section{Characteristic length}

A finite simple graph $G=(V,E)$ defines a finite metric space, where
$d(x,y)$ is the geodesic distance between two vertices $x,y$, the length of the shortest
path connecting the two points. The {\bf characteristic length} $L$ is the expectation
of the distances between different vertices
$\mu(G) =  \frac{1}{n^2-n} \sum_{x \neq y} d(x,y)$,
where $n=v_0=|V|$ is the number of vertices. We have the understanding that the graph 
with only one vertex has zero length $\mu(G)=0$ and that the number $\mu$ is averaged over every connected
component of a graph independently. Unlike the Euler characteristic which 
is a homotopy invariant, the global characteristic length is a metric property as
it  depends on the concrete metric and is only invariant under graph isomorphisms and not under topological 
homeomorphisms \cite{KnillTopology} nor homotopies. Since determining
$\mu$ requires to map out all the distances between any two points, it is natural to ask whether
one can estimate the global length by averaging local properties. Such a relation
has appeared in ``formula 54" of \cite{NewmanStrogatzWatts},
$$  \mu \sim 1+\frac{\log(\delta/n)}{\log(\delta/\delta_2)}  \; , $$ 
where it was derived from generating functions. `Formula 54"
uses the average degree $\delta$ (the average of the degrees $\delta(x)$)
and the average $2$-nearest neighbors $\delta_2$ (the average of the size $\delta_2(x)$ 
of the spheres of radius $2$). Because $\delta=2|E|/n$ by Euler's handshaking lemma, we know that $\delta/n$
agrees with the {\bf edge density} $\epsilon(G)=|E|/\B{n}{2}$. Since $\log(2^{dim(x)-1} \delta/\delta_2)$
measures a relation between volumes of spheres of distance $r$ and $2r$, it can intuitively be
thought of a scalar curvature, the flat case meaning the volume
$\delta_2$ of the sphere of radius $2$ being larger than $2^{\iota(x)-1}$ times the volume  $\delta$ of the sphere
of radius $1$ which has dimension $\iota(x)-1$, where $\iota(x)$ denotes the dimension of the vertex $x$.
The global characteristic length is according to "formula 54"
an "edge density/curvature" relation which is intuitive because distances are small 
in spaces of positive curvature and the fact that if a material has large edge density, it 
allows fast travel. To summarize, "formula 54" relates the edge density $\epsilon$, 
the Hilbert action $\eta$ and the characteristic length $\mu$ as
$$  \mu \sim 1+\frac{\log(\epsilon)}{\log(\eta)}  \; , $$ 
where 
$$ \eta = \frac{1}{n} \sum_{x \in V} \log(\frac{\delta(x)}{\delta_2(x)})  \;  $$
is the average of shifted scalar curvature $s(x) = \log(\delta_2/\delta)$. \\ 

Empirically, we also see a close relation between $\mu$ and $\log(1/\nu)$, where $\nu$ is the mean 
cluster coefficient as defined by \cite{WattsStrogatz}. We have mentioned this in 
\cite{KnillMiniatures} (see also \cite{KnillAffine,KnillOrbital3}). Clustering is a notion 
which is related to transitivity in sociology \cite{WassermanFaust}. We see in many natural networks
that these two quantities $\mu$ and $\log(\nu)$ grow linearly with $\exp(n)$ with a similar 
order of magnitude. Since $C(x)$ can be expressed integral geometrically
using lengths, it can be pushed to natural metric spaces like compact Riemannian 
manifolds or metric spaces with fractal dimension. \\

Why is the characteristic length interesting? In physics, $d(a,b)$ is minimized by geodesics
connecting $a$ with $b$ so that $\mu$ averages over all possible Lagrangian actions of 
paths between any two points. General relativity builds on two variational pillars: one is 
the Hilbert action $\eta$, the average scalar curvature, the other is the 
geodesic variational problem to minimize geodesic length between two points. 
The first tells how matter determines geometry, the second 
describes how geometry influences matter. ``Formula 54" indicates a statistical correlation between
Hilbert action, edge density and characteristic length. Since scalar curvature $S$ appears in the expansion
$|\exp_x(B_r(x)|/|B_r(x)| = 1-S r^2/(6({\rm dim}+2)) + \dots$, we can express this without 
referring to Euclidean space $R^k$ as 
$$  \frac{|\exp_x(B_{r}|}{|\exp_x(B_{2r})|} = 1+r^2 \frac{S}{2(k+2)} +  \cdots $$ 
so that
$$ S \sim  \frac{2(k+2)}{r^2} ( 1+\frac{|\exp_x(B_{2r}|}{|\exp_x(B_{r})|}) 
     \sim  \frac{2(k+2)}{r^2} \log( \frac{|\exp_x(B_{2r})|}{|\exp_x(B_r)|} ) \; . $$
Since $\delta$ and $\delta_2$ are averages, the mean field approximations are 
only expected to be good in some mean field sense.  \\

We would like to know more about the relation between Euler characteristic $\chi$, 
graph density $\epsilon$, dimension $\iota$, complexity $\xi$, clustering $\nu$ 
and characteristic length $\mu$. Also interesting are relations with corresponding notions 
of Riemannian manifolds. 
Besides average length and dimension, complexity makes sense for Riemannian manifolds too, even so it
needs zeta regularized determinants for its definition. 
The mean clustering $\nu$ is proportional to the volume with a proportionality factor which depends 
on the dimension.  For a Riemannian manifold, scaling the space by a factor $n$ scales the 
characteristic length in the same way and the volume by a factor $n^d$. For random networks, 
the global characteristic length
typically grows like $\log(n)$ in dependence of volume $n$ - one aspect of the  ``small world" phenomenon. 
This does not violate geometric intuition at all because dimension grows too with more nodes. 
We have explicit formulas \cite{randomgraph} for the average dimension of a random 
Erd\"os-Renyi graph of size $n$, where edges are turned on with probability $p$. 
For other random graphs like Watts-Strogatz networks or networks generated by random permutations, 
we see similar growth rates. Other type of networks can show slightly different growth rates.
Barabasi-Albert networks are examples, where the growth rate is slower. \\

Related to characteristic length is the {\bf magnitude} $|G|= \sum_{i,j} Z^{-1}_{ij}$, where
$Z_{ij} = \exp(-d(i,j))$ defined by Solow and Polasky. We numerically see that at least for
small vertex cardinality $n$ and connected graphs,
the complete graph has minimal magnitude and the star graph maximal magnitude, a feature which
is shared for many functionals (see Figure~\ref{extrema}). 
Also the magnitude can be defined for more general metric spaces so that one can look for a general metric 
space at the supremum of all $|G|$ where $G$ is a finite subset with induced metric. 
The {\bf convex magnitude conjecture} of Leinster-Willington claims that for convex 
subsets of the plane, $|A|=\chi(A) + p(A)/4 - a(A)/(2\pi)$, where $p$ is the perimeter 
and $a$ is the area.

\section{Local cluster coefficient}

Given a subgraph $H$ of a graph $G$, define the {\bf relative characteristic length} as
$$ \mu(H,G) = \frac{1}{|H| (|H|-1)} \sum_{x,y \in H, x \neq y}  d_G(x,y)   \; . $$
The difference $\nu_H(G) = \mu(H)-\mu(H,G)$ is nonnegative. It is zero if all geodesics 
connecting two points in $H$ remain in $H$. The notion makes sense in any metric 
space equipped with a probability measure. The number $\nu_H(G)$ is a measure for how 
far $H$ is away from being convex within the metric space $G$. The notion depends on a choice
of a probability measure on $G$. On graphs, many fractal sets, spaces on which a Lie group
acts transitively or Riemannian manifolds, there is a natural measure. \\ 

{\bf Examples: } \\
{\bf 1)} For a compact surface $H$ in three dimensional space $G=R^3$ for example,
$\nu_H(G)$ is zero if and only if $H$ is a plane. We understand that the average has been
formed with respect to some absolutely continuous probability measure on the surface $H$ 
which gives positive measure to every open set. \\
{\bf 2)} For a curve in a compact Riemannian 
manifold $G$, the relative characteristic length is $0$ if the curve is a short
enough geodesic. But the relation $\mu(H,G)/\mu(H)$ will decrease eventually to zero for
aperiodic geodesic paths $H$ as the geodesic will accumulate. \\
{\bf 3)} For a region $G$ in Euclidean space, we have $\mu(H,G)=\mu(H)$ if and 
only if $H$ is {\bf convex} in $G$ in the sense that for any two points $x,y \in H$
there is a geodesic in $G$ which is also in $H$. 

Define the {\bf local characteristic length}  as
$$   \mu(x) = \mu(S(x),B(x)) \; , $$
where $S(x)$ is the unit sphere and $B(x)$ is the unit ball of the vertex $x$. We always 
assume that the space is large enough so that the unit ball at every point is convex within $G$ 
in the above sense. 
The quantity $L(x)$ is defined therefore for large enough 
metric spaces equipped with a probability measure $m$ which is nice enough that
it induces measures on spheres $S(x)$ by limiting conditional expectation. Examples
are Riemannian manifolds or graphs. \\

For a finite simple graph $G=(V,E)$, the {\bf local cluster coefficient} is defined as 
$$   C(x) = 2 |E(x)|/(|V(x)|+1)|V(x)|  \; , $$
where $E(x)=V_1(x)$ is the set of edges in the unit sphere $S(x)$ of $x$ and $V(x)=V_0(x)$ 
is the set of vertices in $S(x)$. The {\bf mean cluster coefficient} $\nu(G)$ was defined 
by Watts-Strogatz is the average of $S(x)$ over all $x \in V$. The local cluster
coefficient gives the edge occupation rate in the sphere $S(x)$ of a vertex $x$. 
Other related quantities are the {\bf global cluster coefficient} 
defined as the frequency of oriented triangles $3 v_2/t_2$ within 
all oriented connected vertex triples. The {\bf transitivity ratio} is the ratio $v_2/s_2$
of non-oriented triangles within the class of non-oriented connected vertex triples in $G$. 
We do not look at the later two notions because they are close to the mean cluster
density and because intuition about them is more difficult.  \\

We can define {\bf higher global cluster coefficients} $C_k(G)$ of a graph 
as the fraction $v_k/w_k$, where $v_k$ is the number of $k$-dimensional
simplices $K_{k+1}$ in $G$ and $w_k$ the number of connected $k$-tuples of 
vertices. Of course, $C_1(G)$ is the global clustering coefficient
and $C_1(S(x))$ is the local clustering coefficient at a point $x$. 
One could define higher local clustering coefficient $C_k(x)$ as the 
global clustering coefficient $C_k(S(x))$. There are also higher dimensional 
characteristic lengths: define $d_k(x,y)$ as the distance between two 
$k$-dimensional simplices $x,y$, where the distance is the smallest
$l$ such that we can connect $x,y$ with a sequence $x_0=x,x_1,x_2,\dots,x_l$ 
of overlapping $k$ simplices $x_k$. Of course $d_1(x,y)$ is the usual
geodesic distance and for geometric graphs of dimension $d$ all distances
$d_k$ are essentially the same. For a triangularization of a Riemannian 
manifold, where every sphere is one-dimensional cyclic graph, the distance
between two triangles which overlap is $1$.
We mentioned these higher clustering coefficients and higher characteristic lengths
because we believe they could be used to get a closer relation with 
inductive dimension which also uses the entire spectrum of higher dimensional 
simplices and not only zero dimensional vertices and one dimensional edges. 

\section{Cluster coefficient and local length}

While the following observation is almost obvious, we are not aware
that it has been noticed anywhere already. The result allows to write the 
local cluster coefficient in terms of the relative characteristic length $L(x)$ 
of the unit sphere $S(x)$ within the unit ball $B(x)$. This will allow
us to push the notion of local clustering coefficient to other metric spaces 
equipped with natural measures. Just define then $C(x) =2-L(x)$ there. 

\begin{lemma}[Cluster-Length-Lemma]
$L(x) = 2-C(x)$. 
\end{lemma}

\begin{proof}
By definition, the distance function takes only two different positive values in the ball $B(x)$. 
The first possibility is ${\rm dist}(x,y)=1$ which is the case if one of the vertices is the center.
The second possibility is $d(x,y)=2$ if both $x,y$ are on the sphere. 
If $d$ is the degree of the vertex $x$ then $B(x)$ has $(d+1)$ vertices and $S(x)$ has
$d$ vertices. The distance $1$ appears $C(x) d(d-1)$ times on the sphere $S(x)$. 
The distance $2$ appears $(1-C(x)) d (d-1)$ times in the sphere $S(x)$. The average is 
$$  (1 \cdot[ C(x) d (d-1) ] + 2 \cdot [ (1-C(x)) d (d-1) ] ) /(d (d-1)) = 2- C(x)  \; .  $$
\end{proof} 


An extremal case is a star graph $S_n$ which has cluster coefficient $0$ and dimension $0$
and where the distance between any two points of the unit sphere is $2$. 
An other extreme case is the complete graph $K_{n+1}$ which has the cluster coefficient $0$
and dimension $n$, where the distance between any two points of the unit sphere is $1$.  \\

The wheel graph $W_n(x)$ is an example in which each point has a $1$-dimensional sphere,
the local cluster coefficient of the center is $2/(n-1)$ and of the other points $2/3$. 
The mean cluster coefficient for $W_n$ is $m(W_n) = (n (2/3) + 2/(n-1) )/(n+1)$. 
The cluster-length-ratio $\lambda=-\mu/\log(\nu)$ is close to $1$. 

\section{Dimension} 

An other important local quantity is the dimension ${\rm dim}(x)$ of a vertex
and the inductive dimension of a graph, the average dimension of its vertices. 
It was defined in \cite{elemente11} as
$$ {\rm dim}(\emptyset)= -1, {\rm dim}(G) = 1+\frac{1}{|V|} \sum_{v \in V} {\rm dim}(S(v))  \; ,  $$
where $S(v)=\{ w \in V \; | \; (w,v) \in E \; \}, \{ e=(a,b) \in E \; | 
 \; (v,a) \in E, (v,b) \in E \; \}$ 
denotes the unit sphere of a vertex $v \in V$. \\

Already on a local level, there can be relations. If the dimension of a point is zero, 
then clearly the local length $L(x)$ is $2$ and the cluster coefficient is zero. And
also $\lambda$ is zero.  \\
We have shown in \cite{randomgraph} 
that the expectation ${\rm E}_p[{\rm dim}]$ on $G(n,p)$ satisfies the recursion 
$$ d_{n+1}(p) = 1+\sum_{k=0}^n \B{n}{k} p^k (1-p)^{n-k} d_k(p) \; , $$
where $d_0=-1$. Each $d_n$ is a polynomial in $p$ of degree $\B{n}{2}$.

\section{Metric spaces}

The characteristic length $\mu$ of a metric space $(X,d)$ with measure $m$
is the expectation of length $d(x,y)$ on $X^2 \setminus D$, where $D$ is the 
diagonal $\{ (x,x) \}$ in $X \times X$. The local length $L(x)$
is the characteristic length of the unit sphere $S(x)$ within the unit
ball $B(x)$. Motivated by the above, we call $C(x)=2-L(x)$ the local cluster coefficient
of the pint in $(X,d)$ and its expectation $m$, the mean cluster coefficient. 
Lets look at the quantity $-\mu r/(\log(\nu))$, where $r$ is the radius of the small ball and
where the volume of the manifold is $1$. These are integral geometric questions. We 
in general assume that $r$ is scaled in such a way that the radius of injectivity is 
larger than $1$ and the unit ball is contractible and convex. For a Riemannian manifold, the
number $C$ is a dimension-dependent constant, the average distance between two points in the
unit sphere. For manifolds, studying $\lambda$ is therefore equivalent than to study the 
characteristic length. The question is however interesting for fractals. One can ask for
example, what the cluster-length ration $\lambda$ is for the Sirpinsky carpet. \\

{\bf Example 1}. For a flat torus $X = R^2/(rZ)^2$ with flat Riemannian geodesic 
distance and area measure, we have 
$$ \nu = \frac{2}{\pi} \int_0^{\pi/2} (2-\frac{2 \tan(\phi)}{\sqrt{1+\tan^2(\phi)}}) \; d\phi 
     = (2-\frac{4}{\pi}) =  0.72676... $$
and $\mu =r \frac{1}{6} \left(\sqrt{2}+\sinh^{-1}(1)\right) 
         =r(\sqrt{2}+{\rm arcsinh}(1))/6 = r \cdot 0.382598...$. Therefore,
$$ \lambda = \mu/(r \log(1/\nu)) = 0.382598/\log(0.72676) = 1.17837 ... \;.  $$

{\bf Example 2}. Also for a three-dimensional flat torus, the clustering 
coefficient $C(x)$ is constant and given as the average distance of two points 
on a sphere, which is $4/3$. The characteristic length on the other hand is $0.480296$,
a numerical integral which we were not able to evaluate analytically. The
quotient is $\lambda = 1.18454$. \\

{\bf Example 3}. For a two-dimensional sphere of surface area $1$,
the characteristic length $\mu$ is $1/2\sqrt{\pi}$ times the characteristic
length $L(S(x))$ of the unit sphere which is intrinsic and
uses the geodesic length within the surface and not from an embedding. 
We measure it in $R^3$ to be $1.57032...$ so that $\mu=0.442979..$ can be computed 
by using that that the geodesic distance between
two points given in spherical coordinates as $(\phi_1,\theta_1),
(\phi_2,\theta_2)$ is given by the Haverside formula
$H(\phi_1-\phi_2) + \sin(\phi_1) \sin(\phi_2) H(\theta_1-\theta_2)$
where $H(x) = \sin^2(x/2)$ is the Haverside function. Random 
points on the sphere can be computed with the uniform distribution in $\theta$
and the $\arccos$ distribution in $\phi$. Assuming the same $C(x)$ value as in the plane 
(which uses that near a specific point we can replace the sphere with its
tangent space) and get $\lambda = 1.36 \dots $. 

\section{Cluster-length ratio}

Dimension plays a role for characteristic length. 
We measure experimentally that the global length-cluster ratio quantity 
$$ \lambda = -\mu/\log(\nu) \;     $$
is correlated to dimension for Erd\"os-Renyi graphs: 

\begin{figure}
\scalebox{0.21}{\includegraphics{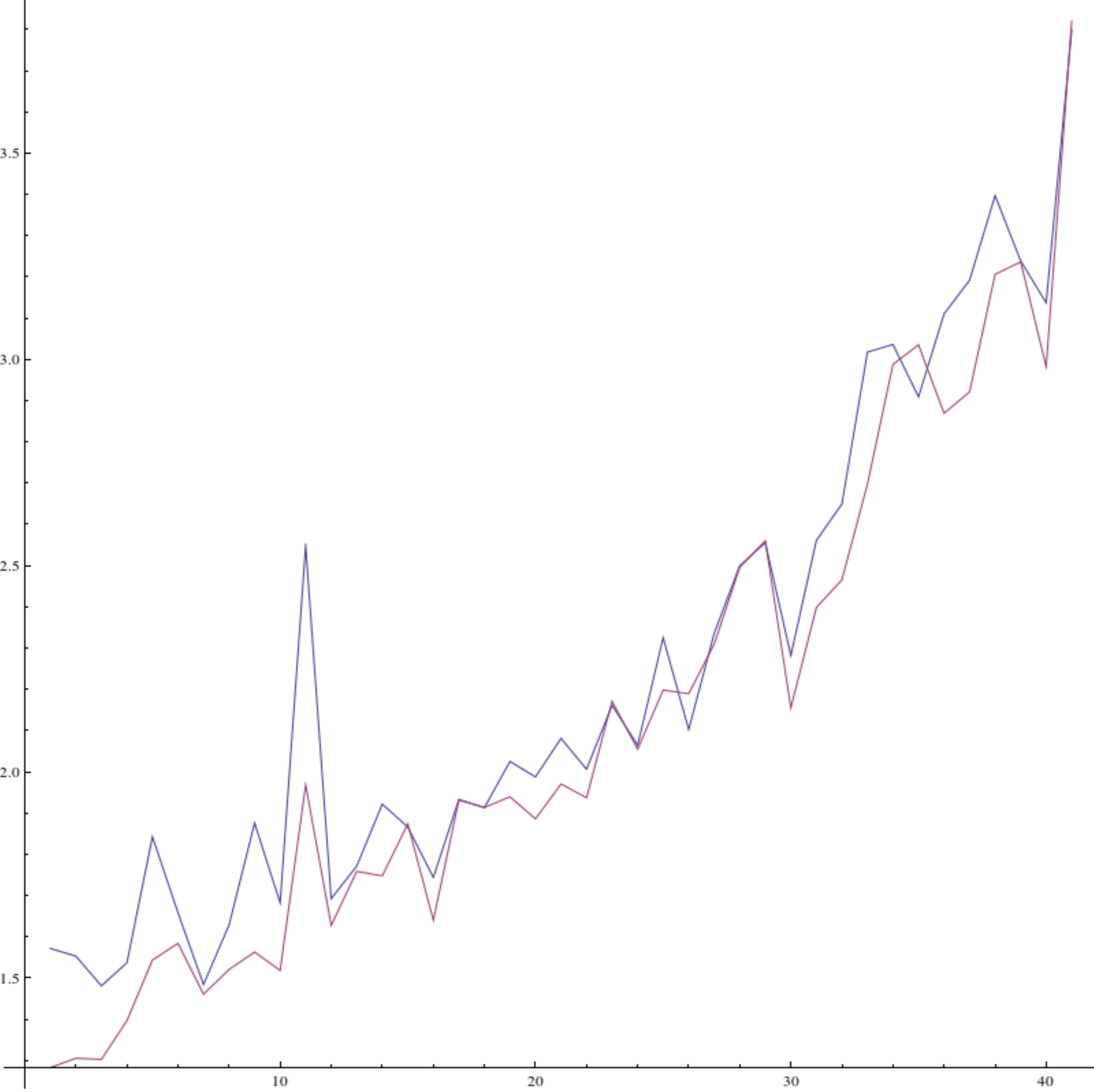}}
\scalebox{0.21}{\includegraphics{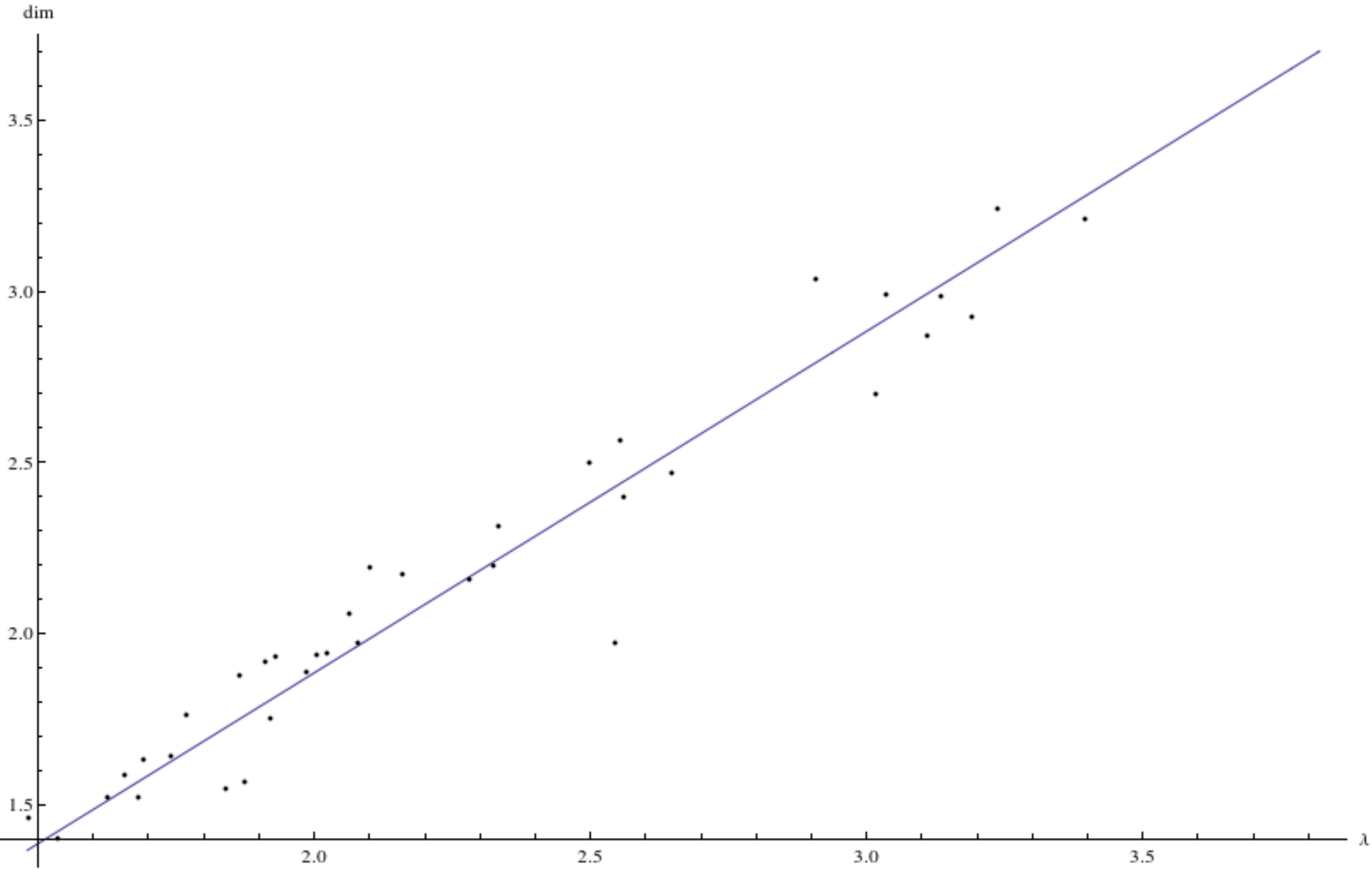}}
\caption{
The plot of the $\lambda(p)$ and $\iota(p)$ in the Erdoes-Renyi case
for probabilities $p$ between $0.3$ and $0.7$, where $n=15$. The graphs 
indicate a clear correlation between cluster-length ratio and dimension. 
The second plot shows the two quantities together with a regression line. }
\end{figure}

\begin{figure}
\scalebox{0.12}{\includegraphics{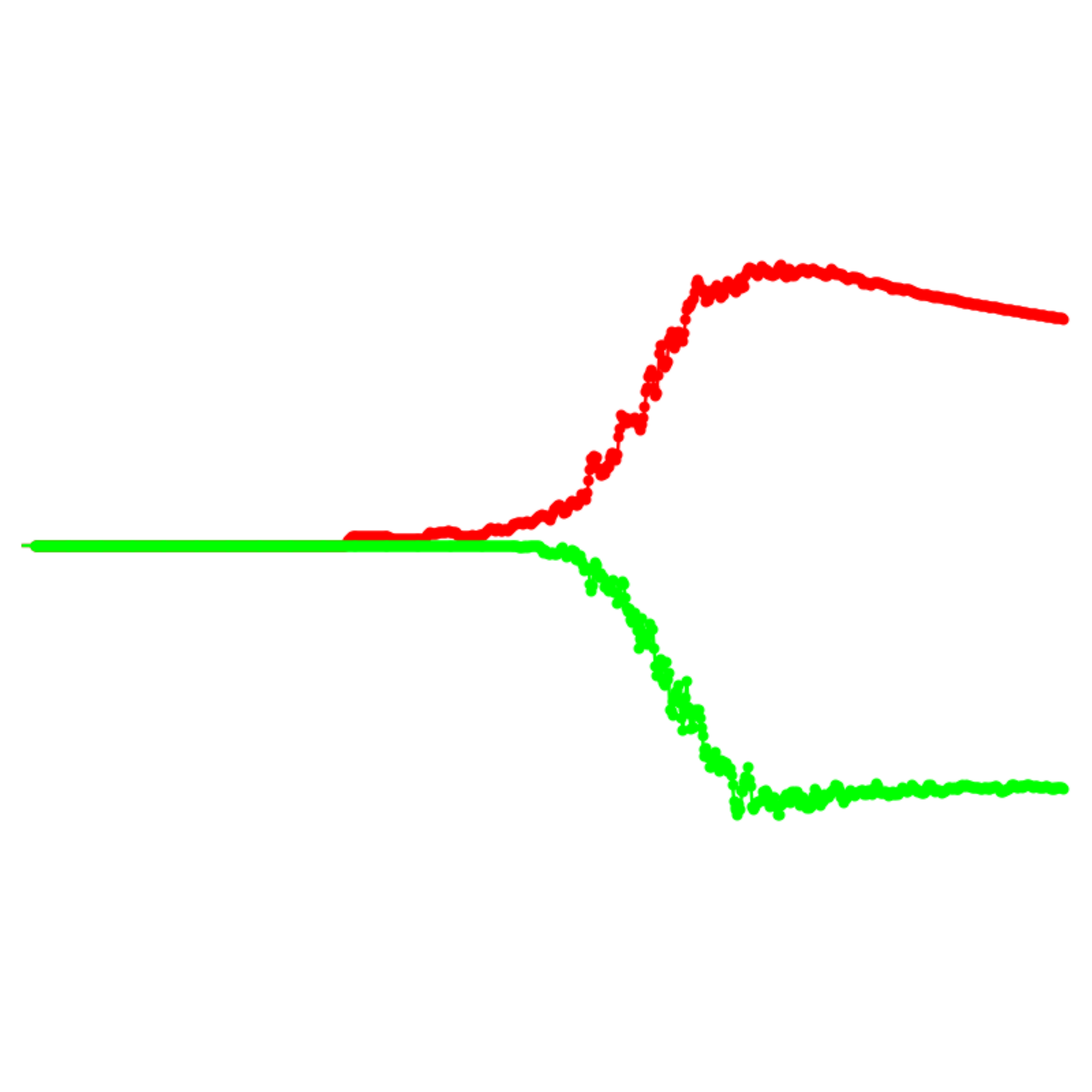}}
\scalebox{0.12}{\includegraphics{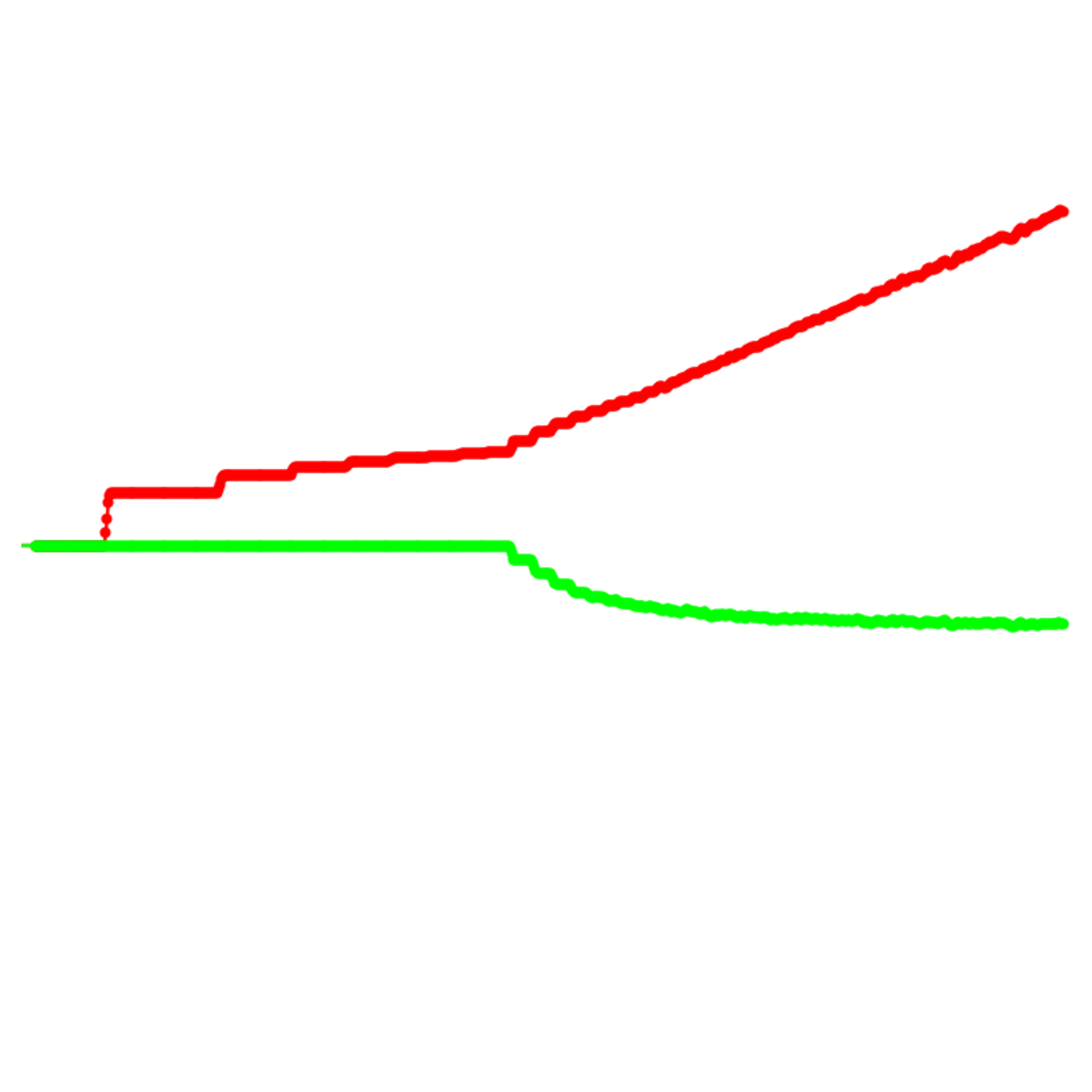}}
\scalebox{0.12}{\includegraphics{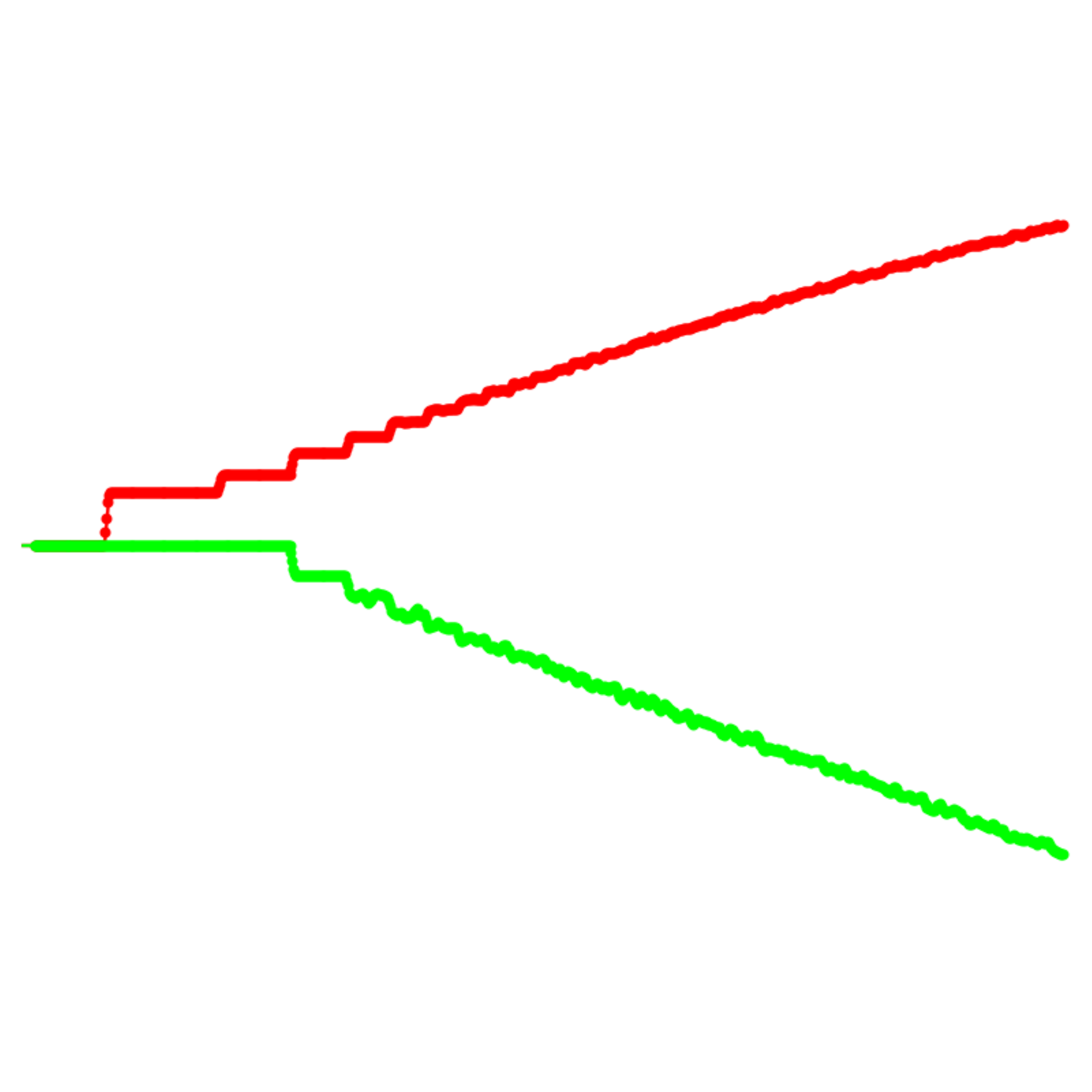}}
\scalebox{0.12}{\includegraphics{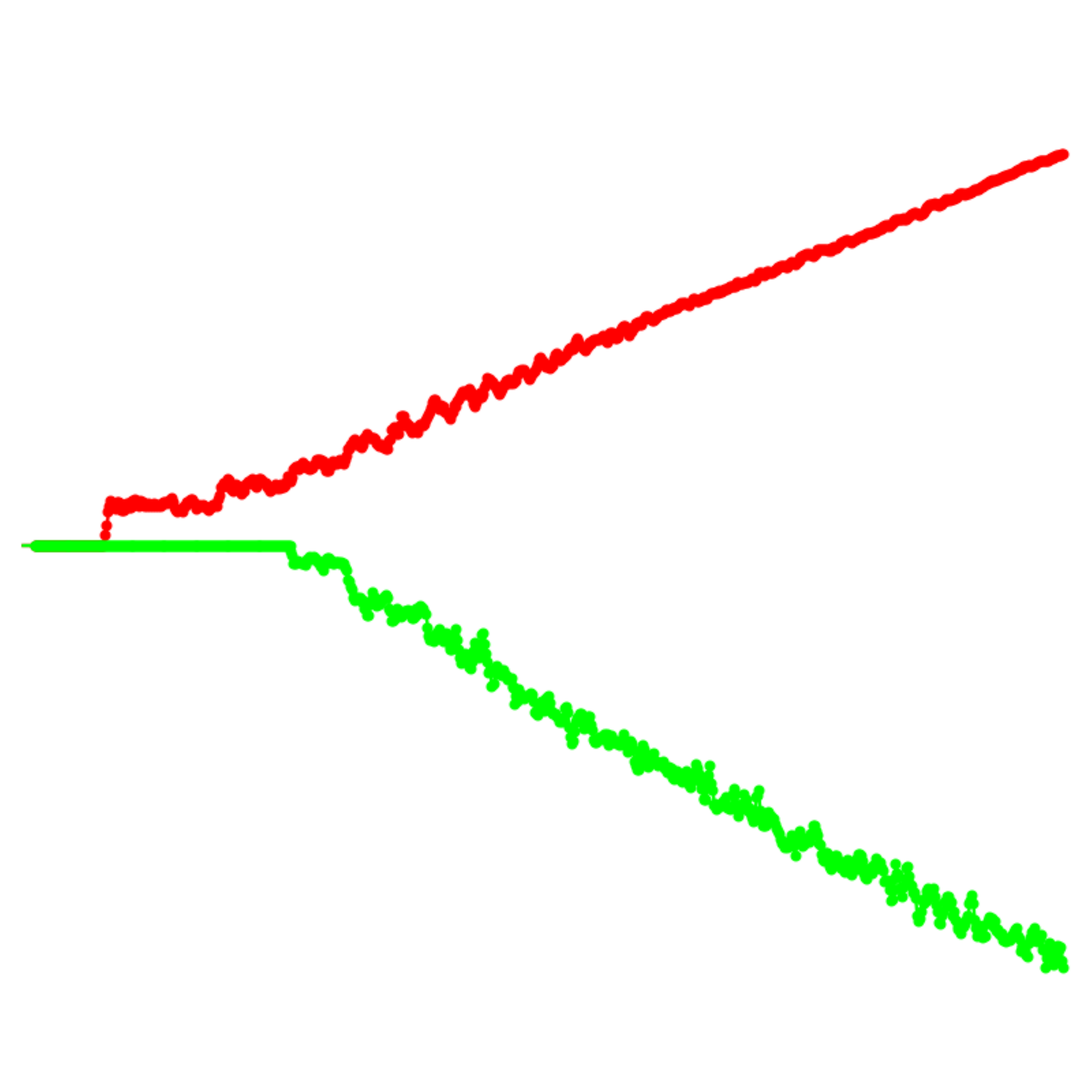}}
\scalebox{0.12}{\includegraphics{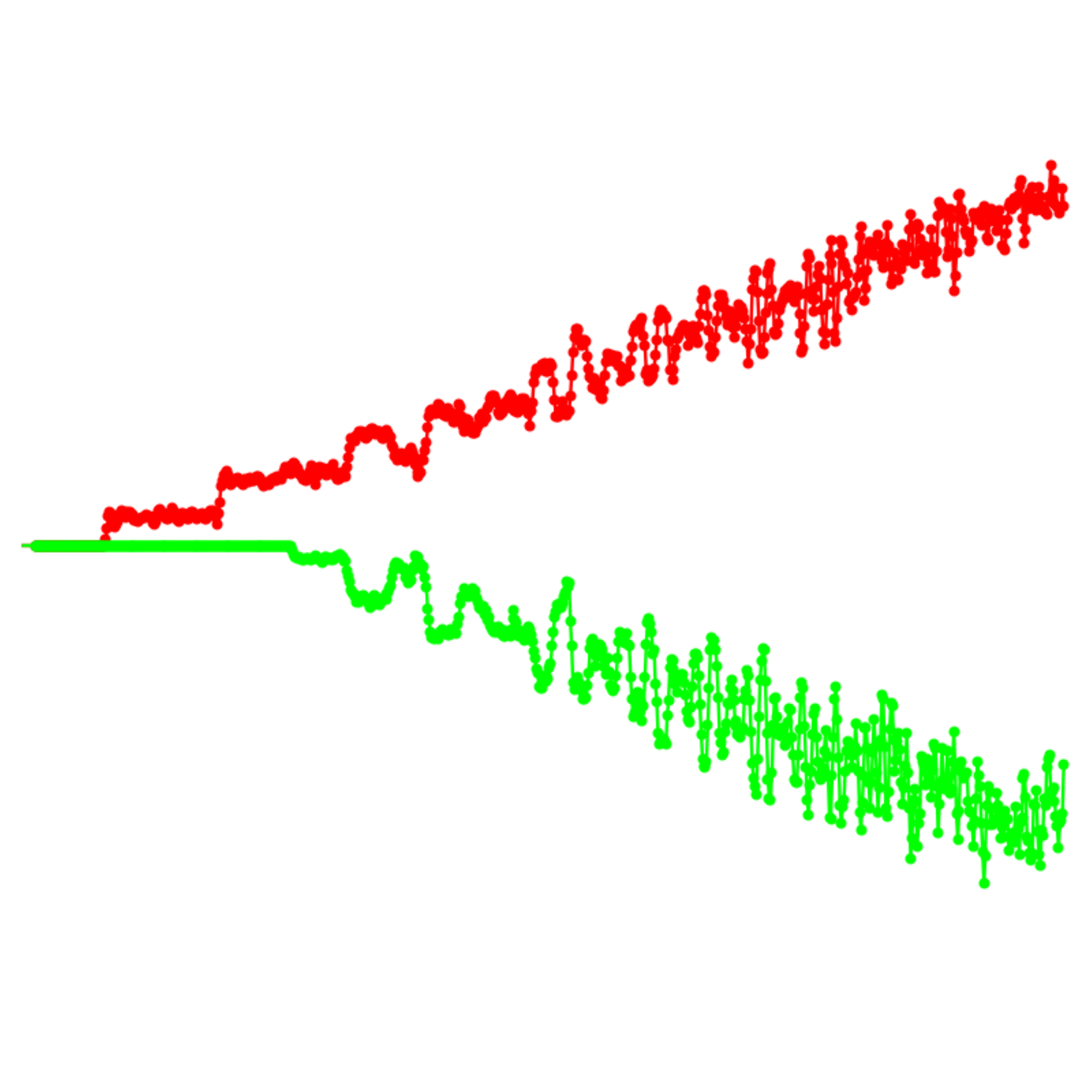}}
\scalebox{0.12}{\includegraphics{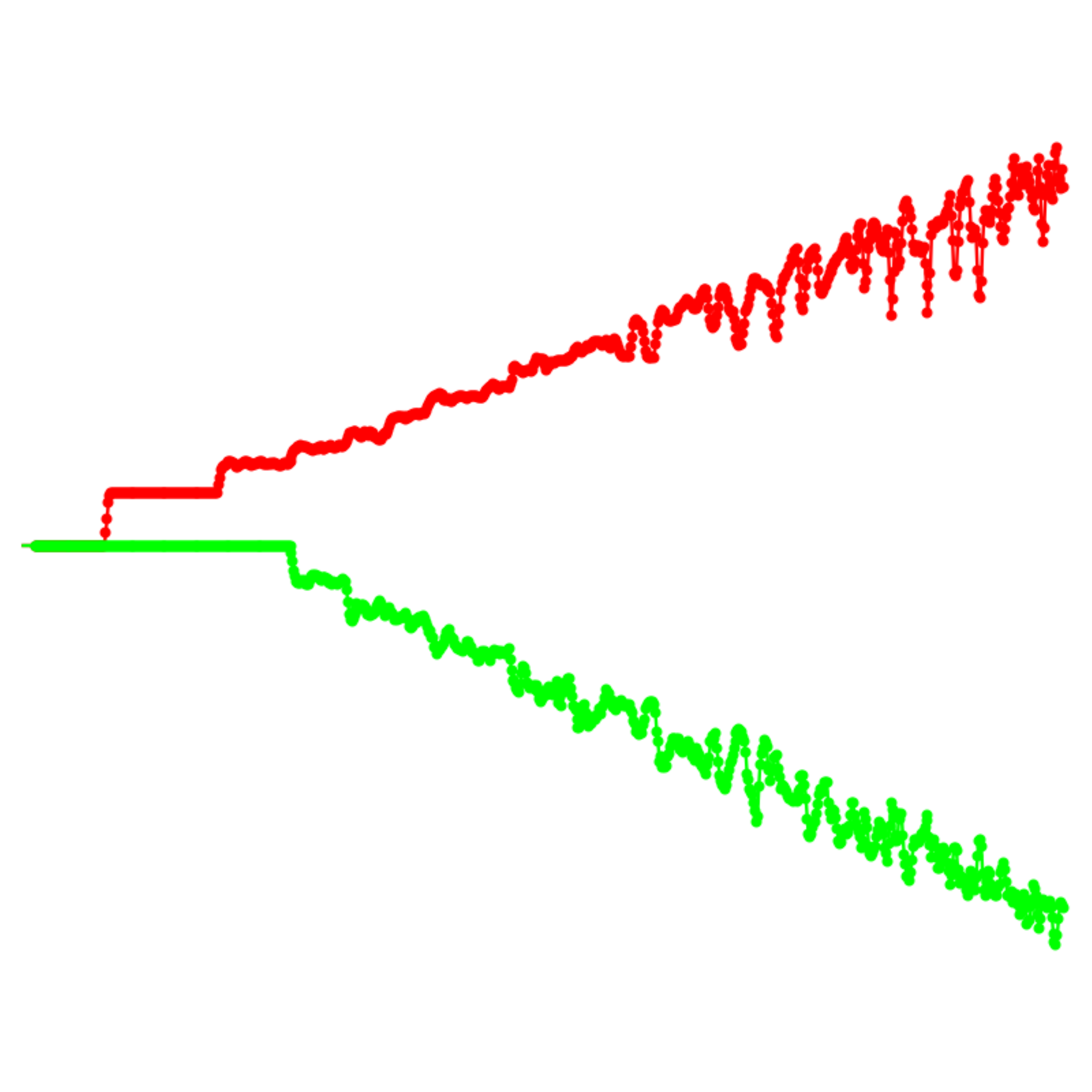}}
\caption{
The characteristic length $\mu$ (red), the log of the cluster coefficient $\log(\nu)$ (green)
and the average vertex degree $d$ (blue) shown together as a function of the number of vertices
in the network. $n$. We plot with a logarithmic
scale in $n$ so that logarithmic growth is seen linearly.
We first use the Erdoes-Renyi probability space (where each edge is turned on
with probability $p=0.1$), then for Watts-Strogatz for $k=4, p=0.1$ and then for 
Barabasi-Albert \cite{Barabasialbert}
networks. In the second row, we see first the case of two quadratic maps \cite{KnillOrbital3}
then two two random permutations and finally a case with correlated
generators, where the clustering is extremely small. 
}
\end{figure}

Here is some intuition, why the limit should exist: the cluster coefficient is related to the 
existence of triangles in the graph. For orbital networks \cite{KnillOrbital1,KnillOrbital2,KnillOrbital3}
defined by polynomial maps the number of triangles is bounded by a constant $C$. This implies that 
$\nu(G) \leq C/n$ and $\log(\nu(G)) \leq \log(C)-\log(n)$ for orbital networks, we should get that
$\log(\nu(G))/\log(n)$ has a finite interval as accumulation points. 
To show that $\mu(G)$ grows like $\log(n)$, we don't want 
too many relations $T^w=T^v$ with different words $w,v$. We call this a collision. 
If  $d$ is the number of generators, then, if there were no collisions, the relation
$d^\mu=n$ holds. With $C_2$ double collisions, assuming no triple collisions, we have
the relation $d^\mu-C=n$ and so $\mu = \log(n+C)/\log(d)$. Together with $-\log(\nu) = \log(n)-\log(C)$
we have $\gamma = \log(n+C)/(\log(d) (\log(n)-\log(C_2(n))))$. 
So, if we can show $C_2$ to be of the order $\log(n)$ and the number of triangles to be 
of the order $o(n)$, and triple collisions are rare, then we should be able to prove that the 
limit exists. 

\section{Classes of networks}

The space $E(n,p)$ of all graphs on a vertex set with $n$ nodes, where every node is turned on 
with probability $p$ is a probability space. The limit 
\begin{eqnarray*} 
 \label{lengthcluster}
  \lambda = -\mu/\log(\nu) 
\end{eqnarray*} 
for $n \to \infty$ exists almost surely. We see that the value is
close to $r/{\rm dim}$, where $r$ is the radius of the graph and ${\rm dim}$ is
the dimension of the network. \\

The quotient~(\ref{lengthcluster}) is interesting because 
$\mu$ is a global property and $\nu$ is the average of a local property. 
Intuitively, such a relation is to be expected because a larger $C(x)$ allows to tunnel 
faster through a ball $B(x)$ and allows for shorter paths. If the limit exists, 
then $\mu= \lambda \log(\nu)$.  
Knowing $\lambda$ is important because the characteristic length is more costly to compute 
while the clustering coefficient $C$ is easier to determine as a simple average 
of local quantities. 
To allow an analogy from differential geometry, we could compare $C(x)$ with curvature, because a 
metric space with larger curvature has a smaller 
average distance between two points on the unit sphere.  \\

We could look at graphs with a given dimension and volume and minimize the average path
length between two points among all graphs. It is a long shot but one can ask
whether there is the relation between graphs minimizing $\mu$ and 
graphs minimizing Euler characteristic $\chi$. We can only explore this so far for very 
small graphs. The reason for asking this is that Euler characteristic can also been seen 
as an average of scalar curvature and therefore a quantized Hilbert action 
\cite{eveneuler}.

\section{Related questions}

Variational problems on graphs usually need some constraints because  the functionals are often
trivial without restrictions. We can restrict the 
number of vertices or edges and look at the maximum or minimum on that space. 
More generally, we can use a Lagrange type problem and look at all the graphs for which one
functional is constant and extremize the other on that class. This leads to more
questions and most of them seem not have been studied. Instead of restricting to a ``level surface"
we can also look at the functionals on an equivalence class of graphs. One interesting 
example is to look at homotopy as an equivalence relation. A homotopy step $G \to G'$ is given by
choosing a contractible subgraph $H$ of $G$  and connect each vertex of $H$ with a new vertex $v$.
An other homotopy step is the reverse operation: remove a vertex for which the unit sphere is 
contractible. The notion of contractible if a sequence of homotopy steps transforms it to a one
point graph.  \\

Lets look at the example of minimizing the dimension $\iota(G)$ in a homotopy class. The 
homotopy class of a circle contains  graphs of arbitrary 
large dimension; it contains for example
discretization of a solid torus (dimension 3) or an annulus (dimension 2). \\

We can also find one-dimensional graphs homotopic to the circle which are not $C_n$.
We can for example attach one dimensional hairs to the circle without changing dimension,
nor homotopy. We have now a new functional $\iota'(G)$ which is the minimal dimension
among all graphs $H$ homotopic to $G$. For a contractible graph, the minimal dimension
is $0$. On the class of graphs homotopic to the circle the minimal dimension 
is $1$ and for all graphs homotopic to an  icosahedron it is $2$. 
A similar modified dimension $\iota'$ can be defined
in the continuum: define the homotopy dimension of a space $M$ as the minimum of
the Hausdorff dimensions of all compact metric spaces $(X,d)$ homotopic to $M$. \\

The question is whether the minimum is always attained by a geometric graph or
smooth manifolds. 

\bibliographystyle{plain}

\end{document}